\documentclass[conference]{IEEEtran}
\IEEEoverridecommandlockouts
\usepackage{cite}
\usepackage{amsmath,amssymb,amsfonts}
\usepackage{algorithmic}
\usepackage{graphicx}
\usepackage{textcomp}
\usepackage{xcolor}

\def\BibTeX{{\rm B\kern-.05em{\sc i\kern-.025em b}\kern-.08em
    T\kern-.1667em\lower.7ex\hbox{E}\kern-.125emX}}

    \usepackage{subcaption}

\usepackage{cases}

\usepackage{amsthm} 
\newtheorem{claim}{Claim}

\usepackage{booktabs} 
\usepackage{makecell} 

\theoremstyle{remark}
\newtheorem{remark}{Remark}
    
\begin{document}

\bstctlcite{IEEEexample:BSTcontrol}
\title{Ambiguity-Resolved 
       Micro-Doppler Construction \\for  
       Asynchronous Bistatic Sensing}

\author{\IEEEauthorblockN{Yiyi Xu, Zhongqin Wang, Kai Wu, and J. Andrew Zhang}
\IEEEauthorblockA{School of Electrical and Data Engineering, University of Technology Sydney, NSW 2007, Australia \\
Emails: yiyi.xu@student.uts.edu.au, \{zhongqin.wang, kai.wu, andrew.zhang\}@uts.edu.au}
}
\maketitle

\begin{abstract}
Integrated sensing and communications (ISAC) is transforming wireless networks into pervasive sensing platforms, yet practical bistatic deployments are hindered by transceiver clock asynchrony, which distorts channel state information (CSI) through random phase fluctuations. {\color{black}CSI-ratio sanitization introduces nonlinear distortion,  limiting its scalability for multiple targets and complicating subsequent delay-domain and angle of arrival (AoA)-domain processing}, whereas cross-antenna conjugate multiplication (CACC) preserves a linear structure but leaves mirror ambiguity and the second-order by-products that affect motion-induced Doppler signatures. 
This paper develops a micro-Doppler construction framework that resolves the mirror ambiguity and suppresses the residual second-order by-products left after differential CACC. First, cyclic differencing suppresses the dominant mirror component. Then, we exploit the fact that the residual second-order terms lie at differenced delay coordinates and do not follow an ordered AoA steering structure to design a lightweight delay-AoA-Doppler transform-domain pipeline that isolates the desired kinematic response while preventing the residual terms from being coherently accumulated at the target bins. Finally, the filtered responses are aggregated into a higher-signal-to-noise-ratio micro-Doppler representation.
Ablation results verify the necessity of residual by-product suppression and the effectiveness of the proposed multidimensional filtering. 
Experiments on the large-scale WiFi gesture dataset show that the proposed ambiguity-resolved micro-Doppler generalizes consistently across four domain factors, achieving accuracies between 90.1\% and 98.2\%, and outperforming representative baselines by about 18 percentage points on average. 
These results show that suppression of residual second-order by-products facilitates unambiguous sensing in asynchronous bistatic ISAC systems.
\end{abstract}

\begin{IEEEkeywords}
Bistatic Sensing, Micro-Doppler, Gesture Recognition, Integrated Sensing and Communications.
\end{IEEEkeywords}

\section{Introduction}
Integrated sensing and communications (ISAC) is increasingly enabling wireless infrastructures to serve as sensing platforms~\cite{masouros2023guest}. In such systems, channel state information (CSI) offers a convenient measurement source to recover motion-induced propagation dynamics~\cite{wu2025isac,wang2025siso}. However, practical bistatic asynchronous transceivers introduce random phase offsets that distort motion kinematics~\cite{wu2024sensing}. 
Current phase-sanitization techniques remain insufficient to fully resolve this issue. For instance, linear fitting cannot fully eliminate timing and carrier offsets. {\color{black}While cross-antenna methods mitigate the random phase term more effectively, they introduce new problems. Specifically,} the CSI-ratio operation {\color{black}is vulnerable to low-amplitude reference measurements} and disrupts the linear structure required for subsequent processing{\color{black}~\cite{9904500}}. Cross-antenna conjugate multiplication (CACC) avoids reference-antenna division and preserves the linear structure, {\color{black}but introduces mirror counterparts~\cite{10737138}, which are difficult to distinguish 
from the desired signal due to the comparable amplitudes.   
Although differential CACC suppresses these mirror components, often-overlooked residual second-order by-products remain, contaminating the desired kinematic signatures.}
Consequently, existing methods struggle to simultaneously achieve robust random phase offset mitigation, multidimensional separability, and unambiguous micro-Doppler construction{\color{black}, which is essential for extracting fine-grained human kinematics in sensing applications, such as gesture and activity recognition.} 

Beyond the asynchronous-phase challenge, cross-domain generalization is another obstacle for wireless sensing. Physics-feature-based methods seek structured motion representations, such as canonical physical projections, motion change patterns, and trajectory-oriented descriptors~\cite{zhang2021widar3,li2020wihf,wu2022toward,zhang2022handgest}. Concurrently, many recent approaches feed CSI amplitudes and phases into neural networks to learn task-oriented representations~\cite{11062116,2025CaoReal,zhao2025crossfi}. Although these methods can achieve competitive performance, the resulting features often lack physical interpretability. Doppler-based representations, by contrast, encode motion direction and velocity evolution and are therefore physically interpretable. {\color{black}However, their generalization across different environments is often limited.} 

To address these challenges, we propose an ambiguity-resolved micro-Doppler construction framework for asynchronous bistatic sensing with a multi-antenna receiver. We first use CACC to obtain a numerically stable, linear-preserving representation, and then apply differencing to suppress the dominant mirror component. The residual second-order by-products are further attenuated through a lightweight delay-angle of arrival (AoA)-Doppler transform-domain pipeline. The delay transform coherently accumulates the motion component into its corresponding delay bin, while the second-order by-products are distributed across different delay bins and therefore behave as incoherent interference across delay bins. The AoA transform further weakens components that do not follow the ordered steering structure. 
The resulting multidimensional responses are finally aggregated into a compact micro-Doppler representation with enhanced per-bin signal-to-noise ratio (SNR)~\cite{chen2019micro}. Experiments on the large-scale Widar3.0 dataset{\color{black}~\cite{zhang2021widar3}} with 105,624 effective samples across diverse domains show strong leave-one-domain generalization, achieving 90.11\%--98.22\% accuracy and consistently outperforming representative feature baselines \cite{zhang2021widar3, li2020wihf, 2025CaoReal}. 


\section{Signal Model and CSI Preprocessing}\label{sec:signalModel_preprocess}
We first introduce the asynchronous bistatic CSI signal model and then present a CACC-based preprocessing for random phase offset removal and dominant mirror suppression.

\subsection{CSI Model}
\label{sec:csi_model}

We consider an asynchronous bistatic orthogonal frequency-division multiplexing (OFDM) sensing system with one transmit antenna and an $M$-element uniform linear array (ULA) at the receiver. The system operates at carrier frequency $f_c$ with wavelength $\lambda$, bandwidth $B$, and $N$ subcarriers with subcarrier spacing $\Delta f = B/N$. For the ULA with element spacing $\Delta d$, the spatial response at the $m$-th receive antenna for the AoA $\theta$ is $\gamma_m(\theta)=e^{-j2\pi (m-1)\Delta d \sin\theta/\lambda}$. 

Within one coherent processing interval (CPI), $K$ CSI snapshots are collected with sampling interval $\Delta t$, and the path parameters are assumed quasi-static over the CPI. Let $\mathbf{k}=[0,1,\ldots,K-1]^{\mathsf T}$ and $\mathbf{n}=[0,1,\ldots,N-1]^{\mathsf T}$ denote the slow-time and subcarrier index vectors. The corresponding Doppler and delay steering vectors for Doppler frequency $f^D$ and propagation distance $d$ are defined as $\mathbf{b}(f^D)=e^{-j2\pi \mathbf{k}f^D\Delta t}$, and $\mathbf{g}(d)=e^{-j2\pi \mathbf{n}\Delta f d/c}$, respectively.

The CSI matrix at the $m$-th receive antenna across $K$ snapshots and $N$ subcarriers is modeled as~\cite{wang2022single}
\begin{equation}
\mathbf{H}_m
=
(\mathbf{A}_m\odot \boldsymbol{\Phi}_m)\odot(\mathbf{S}_m+\mathbf{X}_m)
+\mathbf{W}_m,
\label{eq:Hi_matrix_rewrite}
\end{equation}
where $\odot$ denotes the Hadamard product, $\mathbf{A}_m$ captures the automatic gain control (AGC) through the snapshot gain vector $\mathbf{a}_m\in\mathbb{R}^{K}$, $\boldsymbol{\Phi}_m$ collects the time- and frequency-varying phase offsets induced by bistatic clock asynchrony together with the antenna-dependent static phase offset $\phi_m$, $\mathbf{S}_m$ denotes the static component, $\mathbf{X}_m$ denotes the dynamic component, and $\mathbf{W}_m$ is additive noise. In the subsequent derivation, the aggregate static background is locally approximated by a dominant equivalent reference component. Specifically,
\begin{equation}
\hspace{-.2cm}
\left\{
\begin{aligned}
\mathbf{A}_m &= \mathbf{a}_m\mathbf{1}_N^{\mathsf T},\\
\boldsymbol{\Phi}_m &= e^{-j(\boldsymbol{\varphi}+\mathbf{1}_K\phi_m\mathbf{1}_N^{\mathsf T})},\\
\mathbf{S}_m &= \mathbf{1}_K\big(\gamma_m(\theta_s)\mathbf{g}(d_s)\odot \boldsymbol{\rho}_{m,s}\big)^{\mathsf T},\\
\mathbf{X}_m &= \sum_{p=1}^{P}\mathbf{X}_{m,p}
=\sum_{p=1}^{P}\gamma_m(\theta_p)\mathbf{b}(f_p^D)\mathbf{g}(d_p)^{\mathsf T}\odot \boldsymbol{\rho}_{m,p}.
\end{aligned}
\right.
\label{eq:csi_components}
\end{equation}
Here, $\mathbf{S}_m$ is characterized by equivalent AoA $\theta_s$, path length $d_s$, and frequency-domain gain $\boldsymbol{\rho}_{m,s}\in\mathbb{C}^{N}$. The dynamic component $\mathbf{X}_m$ is the superposition of $P$ motion-induced paths, where the $p$-th path is characterized by Doppler frequency $f_p^D$, path length $d_p$, AoA $\theta_p$, and complex gain $\boldsymbol{\rho}_{m,p}\in\mathbb{C}^{K\times N}$. The cyclic pairwise differencing developed later involves three receive antennas, and hence $M\geq 3$.

\subsection{Random Phase Offset Removal and Mirror Suppression}
\label{subsec:rpr}

To mitigate random phase offsets, we adopt CACC rather than the conventional CSI-ratio operation, because CACC is less sensitive to perturbations when the reference-antenna amplitude is small, as shown by the following claim.

\begin{claim}[Local perturbation stability of CSI-ratio and CACC]
\label{claim:cacc_vs_ratio}
If the perturbation on the reference antenna is locally small, CSI-ratio becomes sensitive as the reference-antenna amplitude decreases, whereas CACC does not exhibit denominator-driven amplification. Hence, CACC has better local perturbation stability than the CSI-ratio in the low-amplitude regime.
\end{claim}

\begin{proof}
Let $\bar{\mathbf H}_{m}$ denote the noise-free component of $\mathbf H_m$, so that $\mathbf H_m=\bar{\mathbf H}_m+\mathbf W_m$. Define the CSI-ratio and CACC as $\mathbf R_{m_1m_2}\triangleq \mathbf H_{m_1}\oslash \mathbf H_{m_2}$, and $\mathbf C_{m_1m_2}\triangleq \mathbf H_{m_1}\odot \mathbf H_{m_2}^{\ast}$, where $\oslash$ denotes the element-wise division, and $(\cdot)^{\ast}$ denotes complex conjugation. 
Since both operations are element-wise, it suffices to consider one entry, with subscripts omitted for brevity: $H_{m_1}=\bar H_{m_1}+w_{m_1}$, $
H_{m_2}=\bar H_{m_2}+w_{m_2}$, where $w_{m_1}\sim\mathcal{CN}(0,\sigma_{m_1}^2)$ and $w_{m_2}\sim\mathcal{CN}(0,\sigma_{m_2}^2)$ are independent additive noise terms. Assuming $|w_{m_2}|\ll |\bar H_{m_2}|$, the first-order perturbations are $\Delta R_{m_1m_2}
\approx {w_{m_1}}/{\bar{H}_{m_2}}
-{\bar{H}_{m_1}w_{m_2}}/{\bar{H}_{m_2}^2}$, and $\Delta C_{m_1m_2}
\approx
\bar{H}_{m_2}^{\ast}w_{m_1}
+
\bar{H}_{m_1}w_{m_2}^{\ast}$. 
Treating $\bar H_{m_1}$ and $\bar H_{m_2}$ as fixed local operating values and taking expectation with respect to the noise perturbations, we obtain
\begin{align}
\mathbb E\!\left[|\Delta R_{m_1m_2}|^2 \right]
&\approx
{\sigma_{m_1}^2}/{|\bar H_{m_2}|^2}
+
{|\bar H_{m_1}|^2\sigma_{m_2}^2}/{|\bar H_{m_2}|^4},
\nonumber\\
\mathbb E\!\left[|\Delta C_{m_1m_2}|^2 \right]
&\approx
|\bar H_{m_2}|^2\sigma_{m_1}^2
+
|\bar H_{m_1}|^2\sigma_{m_2}^2.
\label{eq:claim1_power_short}
\end{align}
The ratio perturbation contains the denominator-driven term $|\bar H_{m_2}|^{-4}$, whereas the CACC perturbation remains polynomial in the signal amplitudes. Therefore, as $|\bar H_{m_2}|\to 0$, the first-order perturbation power of CSI-ratio diverges, while that of CACC remains bounded.
\end{proof}


Motivated by Claim~\ref{claim:cacc_vs_ratio}, we use CACC to remove the antenna-shared random phase term $\boldsymbol{\varphi}$ while preserving a linear structure for subsequent processing. The element-wise conjugate product of antennas $m_1$ and $m_2$ yields
\begin{equation}
\mathbf{H}_{m_1}\odot \mathbf{H}_{m_2}^{\ast}
=
\mathbf{U}_{m_1m_2}
+
\mathbf{V}_{m_1m_2}
+
\mathcal{O}(\mathbf{W})
+
\mathcal{O}(\mathbf{W}^2),
\label{eq:Z_cacc_matrix_revised}
\end{equation}
where
\begin{equation}
\left\{
\begin{aligned}
\mathbf{U}_{m_1m_2}
&=
\boldsymbol{\Gamma}_{m_1m_2}\odot
\big(\mathbf{S}_{m_1}\odot \mathbf{S}_{m_2}^{\ast}\big)\\[2pt]
\mathbf{V}_{m_1m_2}
&=
\boldsymbol{\Gamma}_{m_1m_2}\odot
\big(\mathbf{S}_{m_1}\odot \mathbf{X}_{m_2}^{\ast}
+
\mathbf{X}_{m_1}\odot \mathbf{S}_{m_2}^{\ast}\big)\\
&\quad+
\boldsymbol{\Gamma}_{m_1m_2}\odot
\big(\mathbf{X}_{m_1}\odot \mathbf{X}_{m_2}^{\ast}\big)
\end{aligned}.
\right.
\end{equation}
Here, $\boldsymbol{\Gamma}_{m_1m_2}
\triangleq
e^{-j(\phi_{m_1}-\phi_{m_2})}
(\mathbf{A}_{m_1}\odot \mathbf{A}_{m_2})$ denotes the combined antenna-dependent phase and AGC gain factor, $\mathbf{U}_{m_1m_2}$ denotes the static component, $\mathbf{V}_{m_1m_2}$ collects the motion-related terms, including the desired first-order term, the mirrored term and {\color{black}the second-order by-product}, $\mathcal{O}(\mathbf{W})$ collects first-order noise terms, while $\mathcal{O}(\mathbf{W}^2)$ denotes second-order noise terms, which are neglected in the subsequent derivation under typical operating conditions.

Since the motion-induced component is typically masked by the dominant static component, we estimate the latter after AGC compensation. {\color{black}
Specifically, let $\hat{\mathbf g}_{m_1m_2}\in\mathbb{R}^{K}$ denote the snapshot-wise gain proxy derived from the reported signal strength. We use it to normalize the CACC output along slow time. Under this compensation,} the static term is approximately invariant over slow time and is estimated as $\widehat{\mathbf{U}}_{m_1m_2} = \mathbf{1}_K \widehat{\mathbf{u}}_{m_1m_2}^{\mathsf T}$, where \smash{$\widehat{\mathbf{u}}_{m_1m_2}^{\mathsf T} = \mathbf{1}_K^{\mathsf T} [(\mathbf{H}_{m_1}\odot\mathbf{H}_{m_2}^{\ast}) \oslash
(\hat{\mathbf g}_{m_1m_2}\mathbf 1_N^{T})]/K$}. 
Subtracting this estimate from the gain-normalized CACC result gives \smash{$\widehat{\mathbf{V}}_{m_1m_2} = (\mathbf{H}_{m_1}\odot\mathbf{H}_{m_2}^{\ast}) \oslash
(\hat{\mathbf g}_{m_1m_2}\mathbf 1_N^{T}) - \widehat{\mathbf{U}}_{m_1m_2}$}. 
To cancel factor $\boldsymbol{\Gamma}_{m_1m_2}$, we normalize it by the static estimate as $\mathbf{W}_{m_1m_2} \triangleq \widehat{\mathbf{V}}_{m_1m_2} \oslash \widehat{\mathbf{U}}_{m_1m_2}$. 
Defining the static-normalized motion-induced term as $\tilde{\mathbf{X}}_m \triangleq \mathbf{X}_m\oslash \mathbf{S}_m$, 
this yields
\begin{equation}
\mathbf{W}_{m_1m_2}
\approx
\tilde{\mathbf{X}}_{m_1}
+
\tilde{\mathbf{X}}_{m_2}^{\ast}
+
\tilde{\mathbf{X}}_{m_1}\odot\tilde{\mathbf{X}}_{m_2}^{\ast}.
\label{eq:static_normalize_revised}
\end{equation}
Equation~\eqref{eq:static_normalize_revised} reveals two detrimental terms for motion-induced signal extraction: the conjugate term $\tilde{\mathbf{X}}_{m_2}^{\ast}$ induces mirror ambiguity at $(-f_p^D,-(d_p-d_s))$, while the nonlinear term $\tilde{\mathbf{X}}_{m_1}\odot\tilde{\mathbf{X}}_{m_2}^{\ast}$ generates second-order by-products that distort the desired kinematic signature.

To suppress the dominant first-order mirror component, we apply cyclic antenna pairwise differencing as
\begin{equation}
\begin{aligned}
&\Delta\mathbf{W}_{[m]_M,[m+1]_M} \triangleq \mathbf{W}_{[m+2]_M,[m]_M}^{\ast} 
- \mathbf{W}_{[m+1]_M,[m+2]_M}, \\
& = 
(\tilde{\mathbf{X}}_{[m]_M}
-
\tilde{\mathbf{X}}_{[m+1]_M})
+
\tilde{\mathbf{X}}_{[m+2]_M}^{\ast}
\odot
(\tilde{\mathbf{X}}_{[m]_M}
-\tilde{\mathbf{X}}_{[m+1]_M})
\label{eq:DeltaW_general_revised}
\end{aligned}
\end{equation}
where
\begin{equation}
\left\{
\begin{aligned}
\tilde{\mathbf{X}}_m
&=
\sum_{p=1}^{P}
\frac{\gamma_m(\theta_p)}{\gamma_m(\theta_s)}
\mathbf{b}(f_p^D)\mathbf{g}(r_p)^{\mathsf T}
\odot
\tilde{\boldsymbol{\rho}}_{m,p},
\\
\tilde{\mathbf{X}}_{m_1}\odot \tilde{\mathbf{X}}_{m_2}^{\ast}
&=
\sum_{p,q=1}^{P}
\boldsymbol{\eta}_{m_1,m_2}^{(p,q)}
\odot
\mathbf{b}(f_p^D-f_q^D)
\mathbf{g}(r_p-r_q)^{\mathsf T}.
\label{eq:second_order_exact_revised}
\end{aligned}
\right.
\end{equation}
Here, $[m]_M\triangleq ((m-1)\bmod M)+1$, $\boldsymbol{\eta}_{m_1,m_2}^{(p,q)}
=
{\gamma_{m_1}(\theta_p)}{\gamma_{m_2}^{\ast}(\theta_q)}/({\gamma_{m_1}(\theta_s)}
{\gamma_{m_2}^{\ast}(\theta_s)})
\big(
\tilde{\boldsymbol{\rho}}_{m_1,p}\odot\tilde{\boldsymbol{\rho}}_{m_2,q}^{\ast}
\big)$, $\tilde{\boldsymbol{\rho}}_{m,p}=\boldsymbol{\rho}_{m,p}\oslash \boldsymbol{\rho}_{m,s}$, and $r_p = d_p-d_s$. 

\begin{remark}
Spatial differencing in \eqref{eq:DeltaW_general_revised} removes the first-order mirror term but leaves second-order residuals which do not share the same delay- and AoA-domain structure as the desired motion-induced component. This structural difference motivates the filtering design in the next section.
\end{remark}

\section{Ambiguity-Resolved Micro-Doppler Construction}
\label{section_microdoppler}

While the cyclic pairwise differencing in \eqref{eq:DeltaW_general_revised} suppresses the dominant mirror component, residual second-order by-products remain. To mitigate these residuals, we tailor a multidimensional filtering pipeline that exploits their transform-domain separability from the desired motion-induced response. 
By aggregating the results across the delay and AoA domains, we derive a compact ambiguity-resolved micro-Doppler representation with enhanced per-bin SNR.



\subsection{Delay-Domain Filtering of Residual By-Products}
\label{subsec:rangefft}

We first transform the signal along the subcarrier dimension to obtain a delay-domain representation. This step not only localizes motion-induced energy along the {\color{black}delay axis}, but also separates the desired response from residual second-order by-products that appear at pairwise-difference delay coordinates. For a signal of length $L$ with a padded size $\bar L\ge L$, define zero-padding operator as 
$\mathbf{Z}_{L\rightarrow \bar L}\triangleq [\mathbf{I}_L;0_{(\bar L-L)\times L}] \in \mathbb{C}^{\bar L\times L}$, where $\mathbf{I}_L$ is the $L \times L$ identity matrix and $\mathbf{F}_{\bar L}$ denotes the $\bar L$-point discrete Fourier transform (DFT) matrix. 
Let $D_M(x)$ denote the normalized Dirichlet kernel, defined as
\begin{equation}
D_M(x)\triangleq \frac{1}{M}\sum_{r=0}^{M-1}e^{jrx}
=\frac{\sin(Mx/2)}{M\sin(x/2)}e^{j(M-1)x/2}.
\label{eq:DM_def_revised}
\end{equation}

Assuming the effective scattering gains are approximately frequency-flat within the signal bandwidth, the delay-domain response is
\begin{equation}
\begin{aligned}
 \Delta \mathbf{W}_{m}^{(\mathrm{d})}
&\triangleq
\Delta \mathbf{W}_{[m]_M, [m+1]_M} \mathbf{Z}_{N\rightarrow N_r}^{\mathsf T}
\mathbf{F}_{N_r}^{\mathsf T}/N \\
& \approx
\sum_{p=1}^{P}
\boldsymbol{\beta}_{m}^{(p)}
\odot
\mathbf{b}(f^D_p)\,\mathbf{q}(r_p)^{\mathsf T}
+
\mathbf{R}_{m}^{(\mathrm{d})},
\end{aligned}\label{eq:range_domain_decomp_revised}
\vspace{-.5cm}
\end{equation}
where 
\begin{equation}
   \left\{
\begin{aligned}
\boldsymbol{\beta}_{m}^{(p)}
&=\frac{\gamma_{[m]_M}(\theta_p)}{\gamma_{[m]_M}(\theta_s)}\tilde{\boldsymbol{\rho}}_{[m]_M,p}
- \frac{\gamma_{[m+1]_M}(\theta_p)}{\gamma_{[m+1]_M}(\theta_s)}\tilde{\boldsymbol{\rho}}_{[m+1]_M,p} \\
[\mathbf{q}(r)]_{\ell}&= D_N(-2\pi(r\Delta f/c+\ell/N_r)), \quad \ell = 0,\ldots,N_r-1\\
\mathbf{R}_{m}^{(\mathrm{d})}
&=
\sum_{p,q=1}^{P}
\boldsymbol{\zeta}_{m}^{(p,q)}
\odot
\mathbf{b}(f^D_p-f^D_q)\,\mathbf{q}(r_p-r_q)^{\mathsf T}\\
\boldsymbol{\zeta}_{m}^{(p,q)}
& =
{\gamma_{[m+2]_M}^{\ast}(\theta_q)}/{\gamma_{[m+2]_M}^{\ast}(\theta_s)}
\tilde{\boldsymbol{\rho}}_{[m+2]_M,q}^{\ast}\odot
\boldsymbol{\beta}_{m}^{(p)}
\end{aligned}.
\right.
\label{eqn:delay_fft_term}
\end{equation}

The delay transform in \eqref{eq:range_domain_decomp_revised} suppresses
residual second-order by-products by redistributing them into delay bins different from those of the desired motion-induced response.
Specifically, the target energy is focused around the physical delays $r_p$, whereas the diagonal auto-terms ($p=q$) collapse toward near-zero delay and the off-diagonal cross-terms ($p\neq q$) appear {\color{black}in delay bins corresponding to} $r_p-r_q$. 

\begin{figure*}[t]
	\centering
	\begin{subfloat}[]{
			\includegraphics[width=0.149\linewidth]{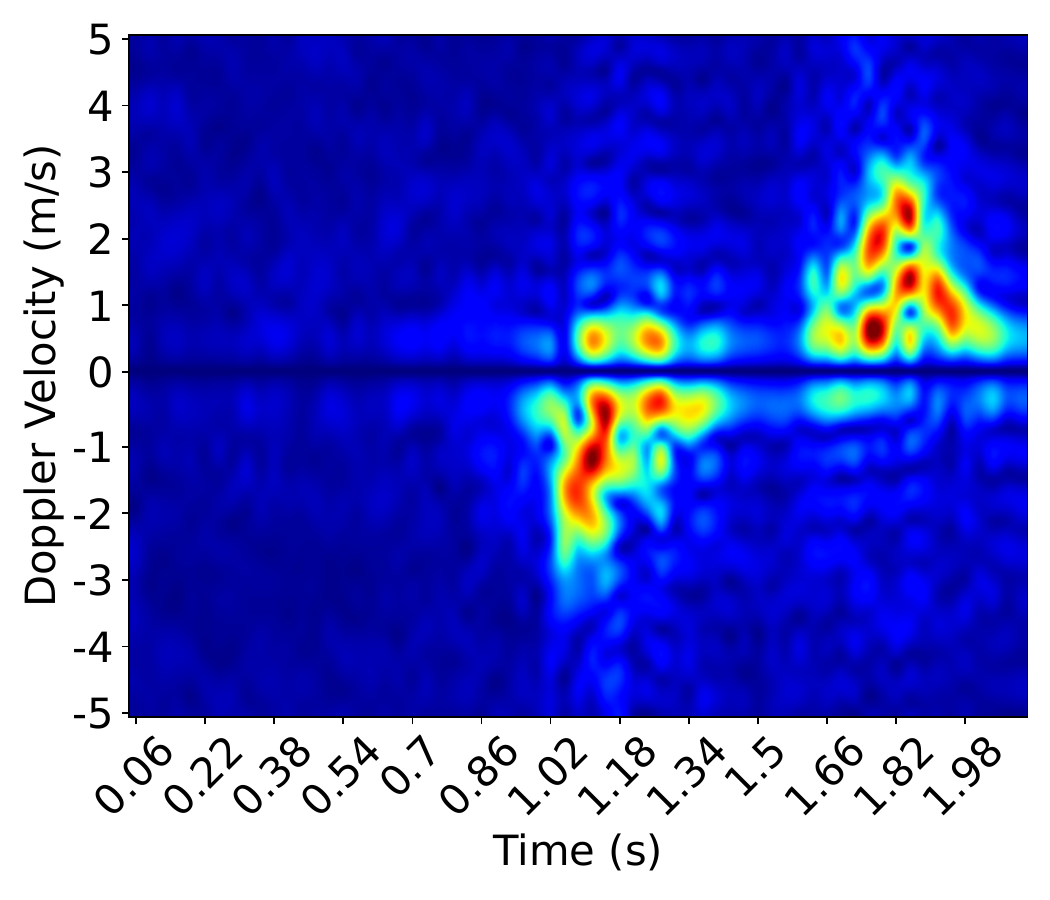}
		\label{subfig:pp_MC_maxdoppler}}
	\end{subfloat}%
	\begin{subfloat}[]{
			\includegraphics[width=0.149\linewidth]{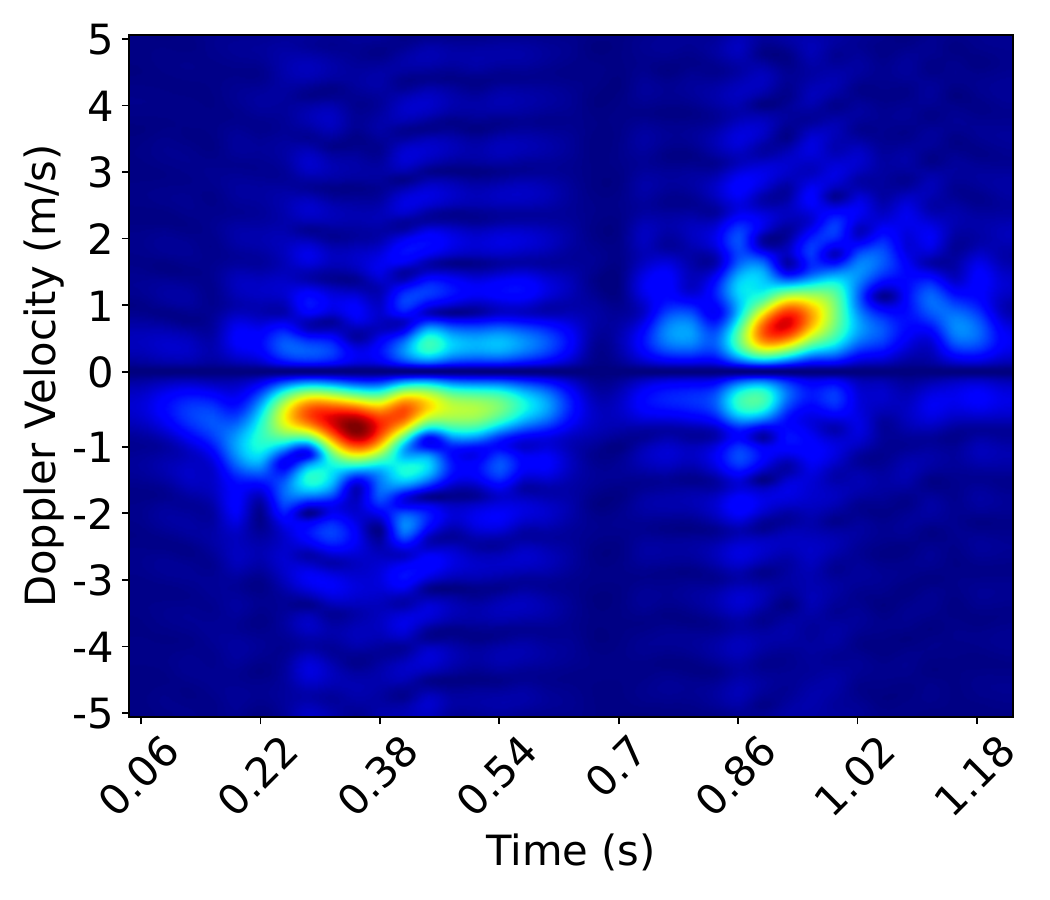}
		\label{subfig:pp_time_doppler_l2o1}}
	\end{subfloat}%
	\begin{subfloat}[]{
			\includegraphics[width=0.149\linewidth]{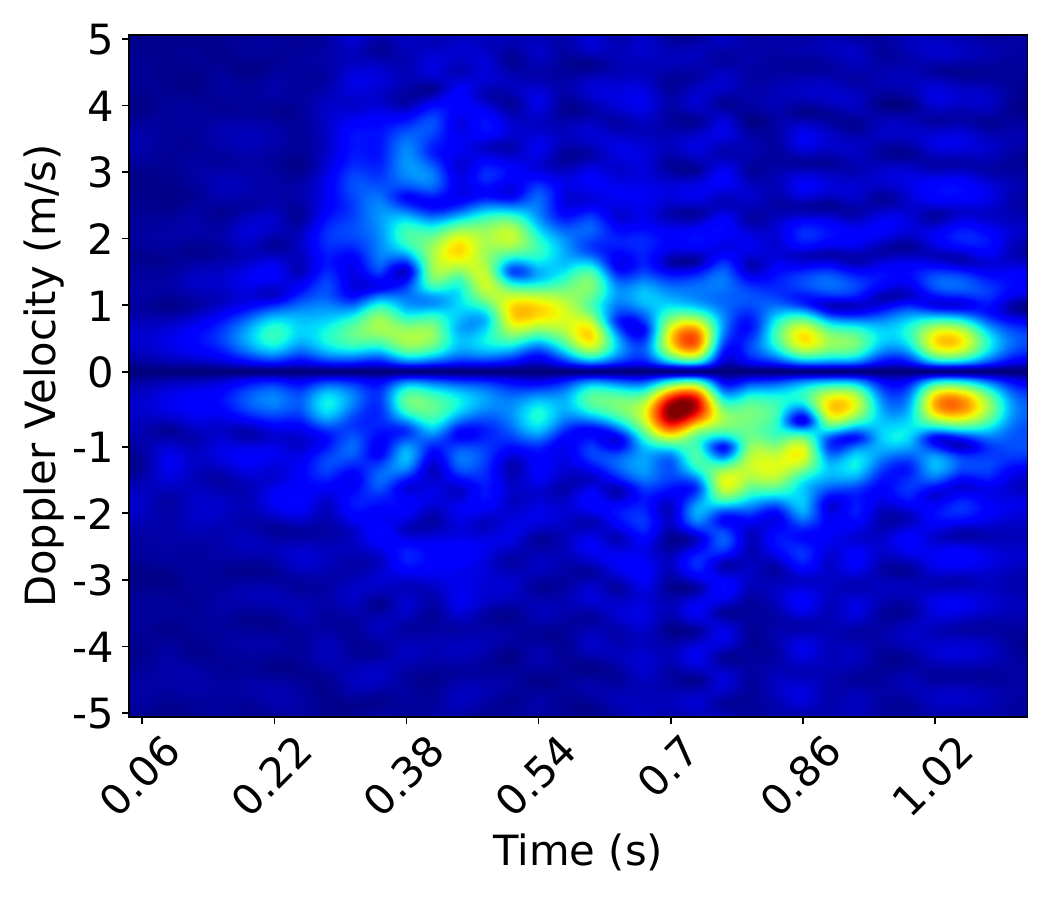}
		\label{subfig:pp_time_doppler_l1o5}}
	\end{subfloat}%
	\begin{subfloat}[]{
			\includegraphics[width=0.149\linewidth]{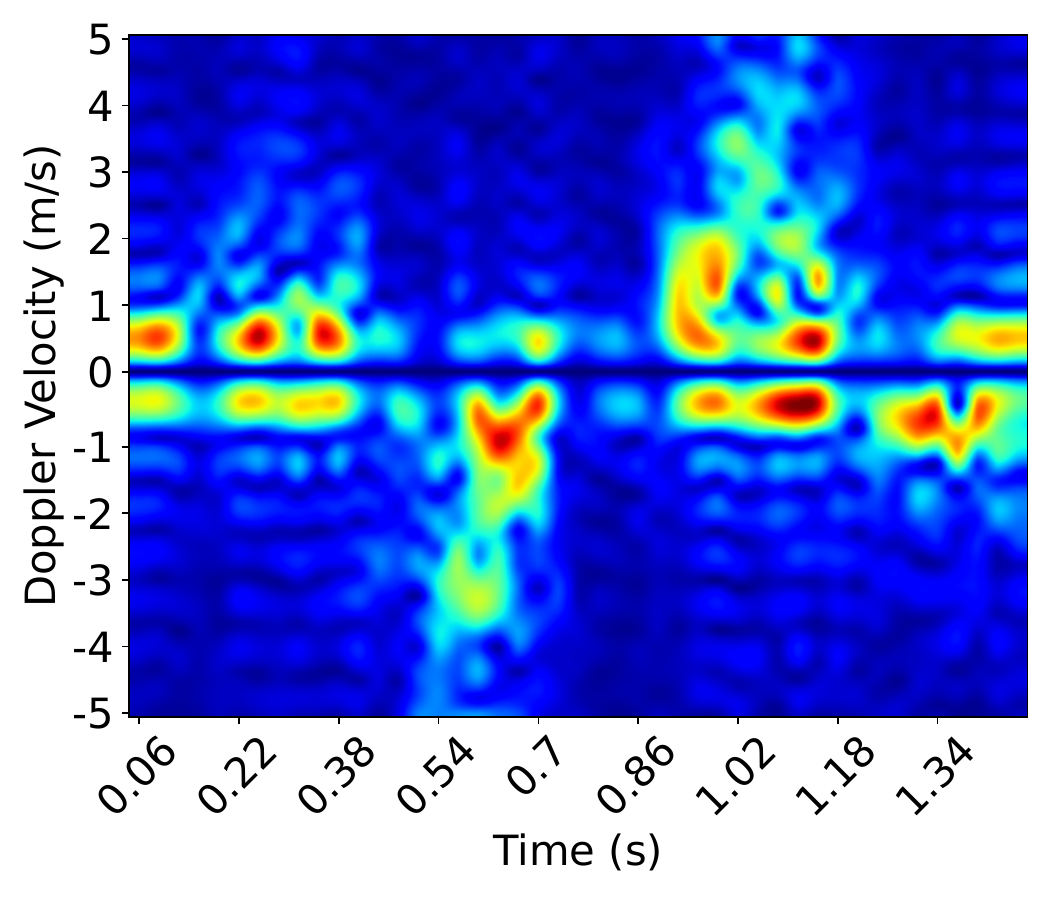}\label{subfig:slide_MC_maxdoppler}}
	\end{subfloat}%
	\begin{subfloat}[]{
			\includegraphics[width=0.149\linewidth]{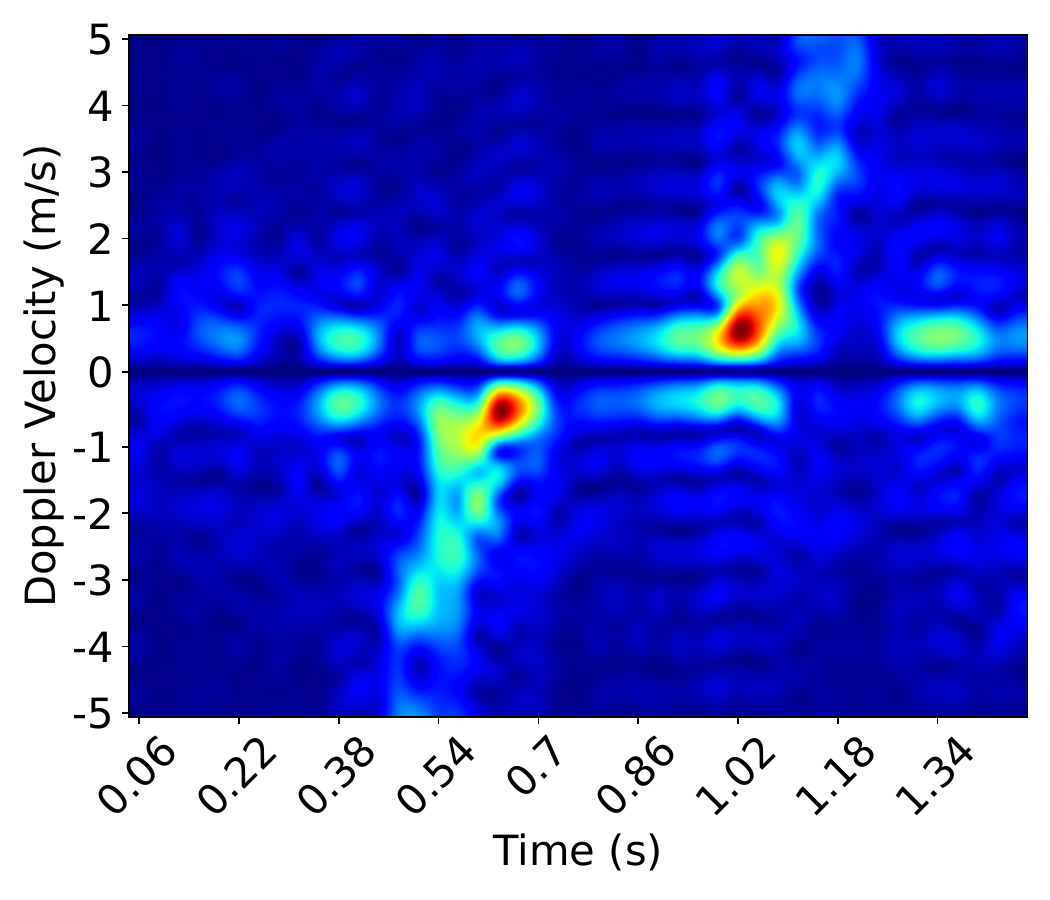}
            \label{subfig:Slide_time_doppler_l2o1}}
	\end{subfloat}%
	\begin{subfloat}[]{
			\includegraphics[width=0.149\linewidth]{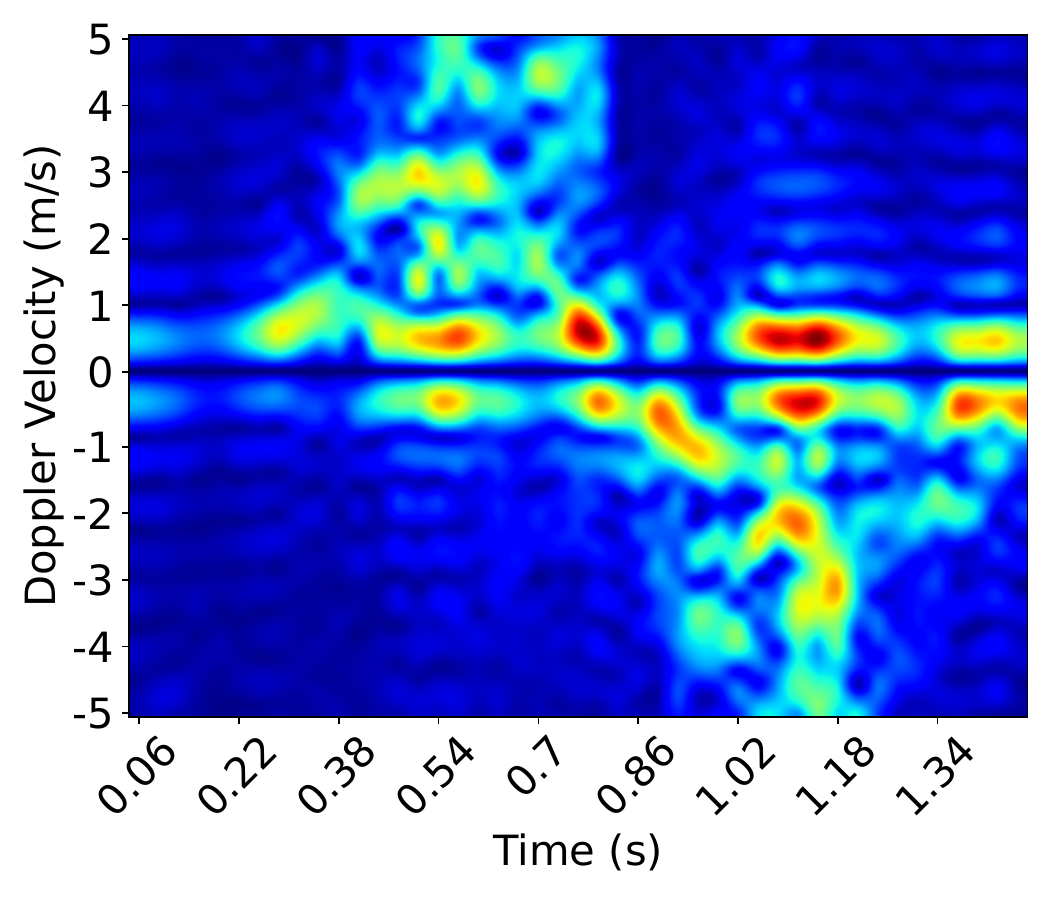}
			\label{subfig:Slide_time_doppler_l1o5}}
	\end{subfloat}
\caption{Impact of sensing geometry on micro-Doppler signatures. (a)-(c) show ``Push \& Pull'' and (d)-(f) show ``Slide''. (a) and (d) Motion along the normal vector, yielding maximum Doppler shifts. 
(b) and (e) Oblique projections resulting in compressed Doppler ridges. 
(c) and (f) Sign flips due to relative motion reversal.}
\label{fig:Slide_time_doppler}
\end{figure*}

\begin{figure*}[t]
    \centering
    \begin{subfloat}[]{
        \includegraphics[width=0.149\linewidth]{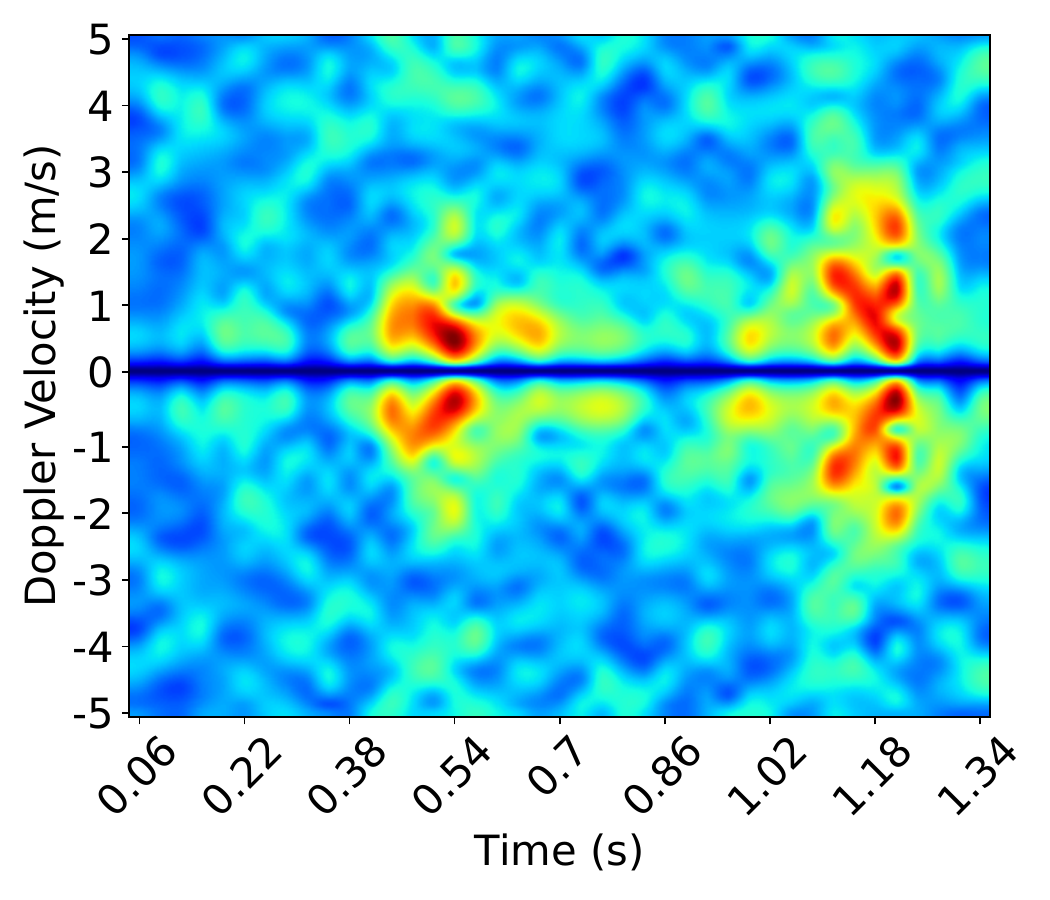}
        \label{subfig:pp_DopplerSpectrum}}
    \end{subfloat}%
    \begin{subfloat}[]{
        \includegraphics[width=0.149\linewidth]{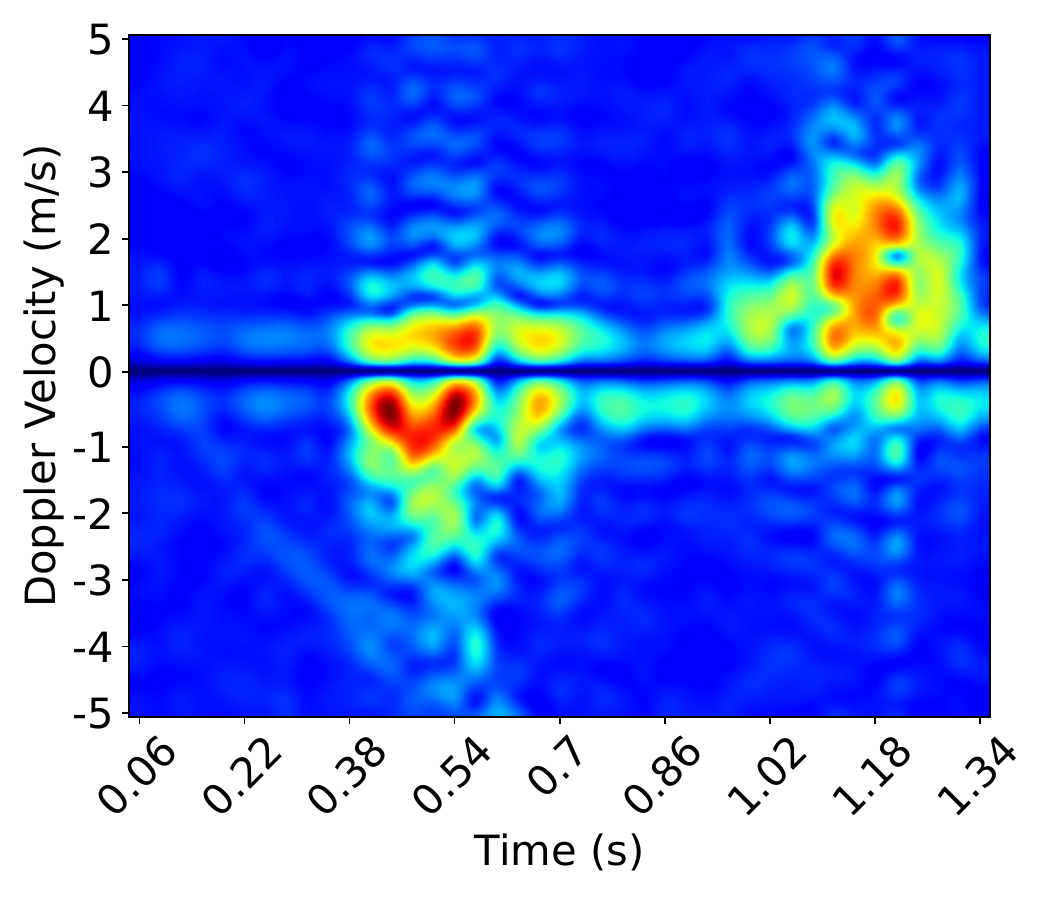}
        \label{subfig:pp_time_doppler}}
    \end{subfloat}%
    \begin{subfloat}[]{
        \includegraphics[width=0.149\linewidth]{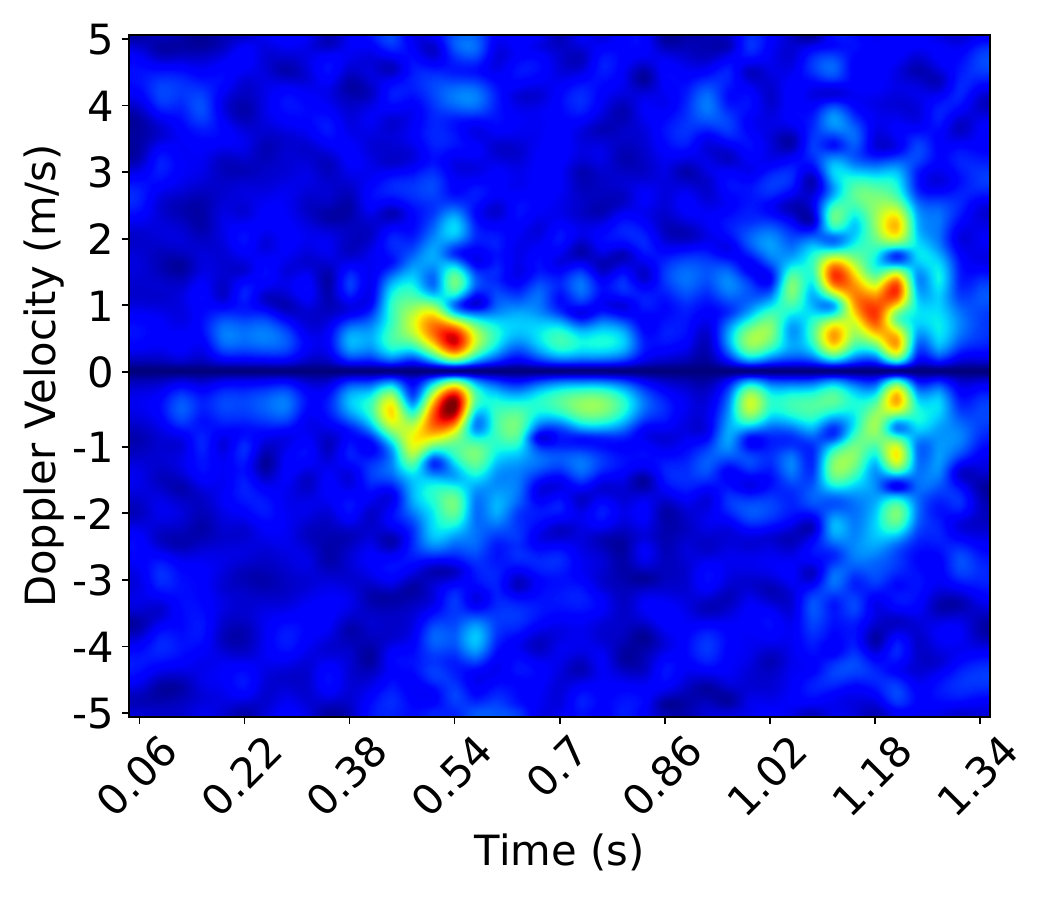}
        \label{subfig:pp_mirror_time_doppler}}
    \end{subfloat}%
    \begin{subfloat}[]{
        \includegraphics[width=0.149\linewidth]{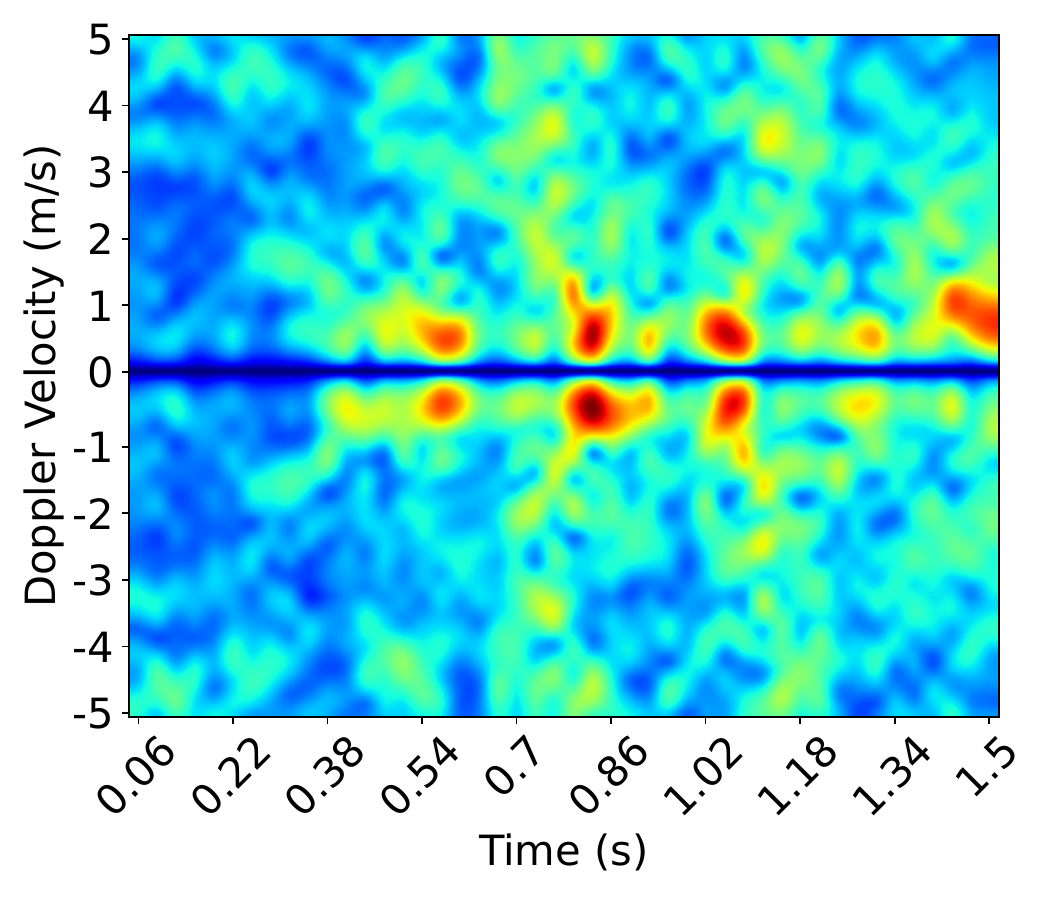}
        \label{subfig:slide_mirror_DopplerSpectrum}}
    \end{subfloat}%
    \begin{subfloat}[]{
        \includegraphics[width=0.149\linewidth]{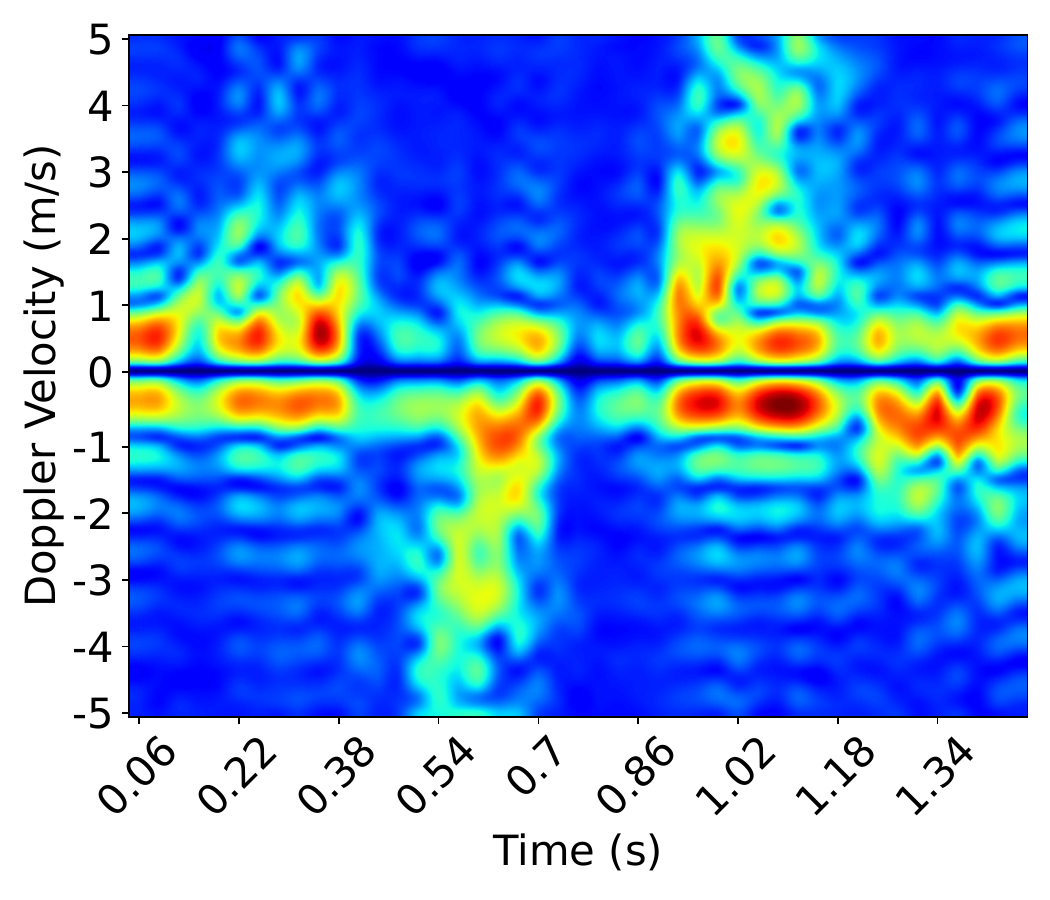}
        \label{subfig:Slide_time_doppler}}
    \end{subfloat}%
    \begin{subfloat}[]{
        \includegraphics[width=0.149\linewidth]{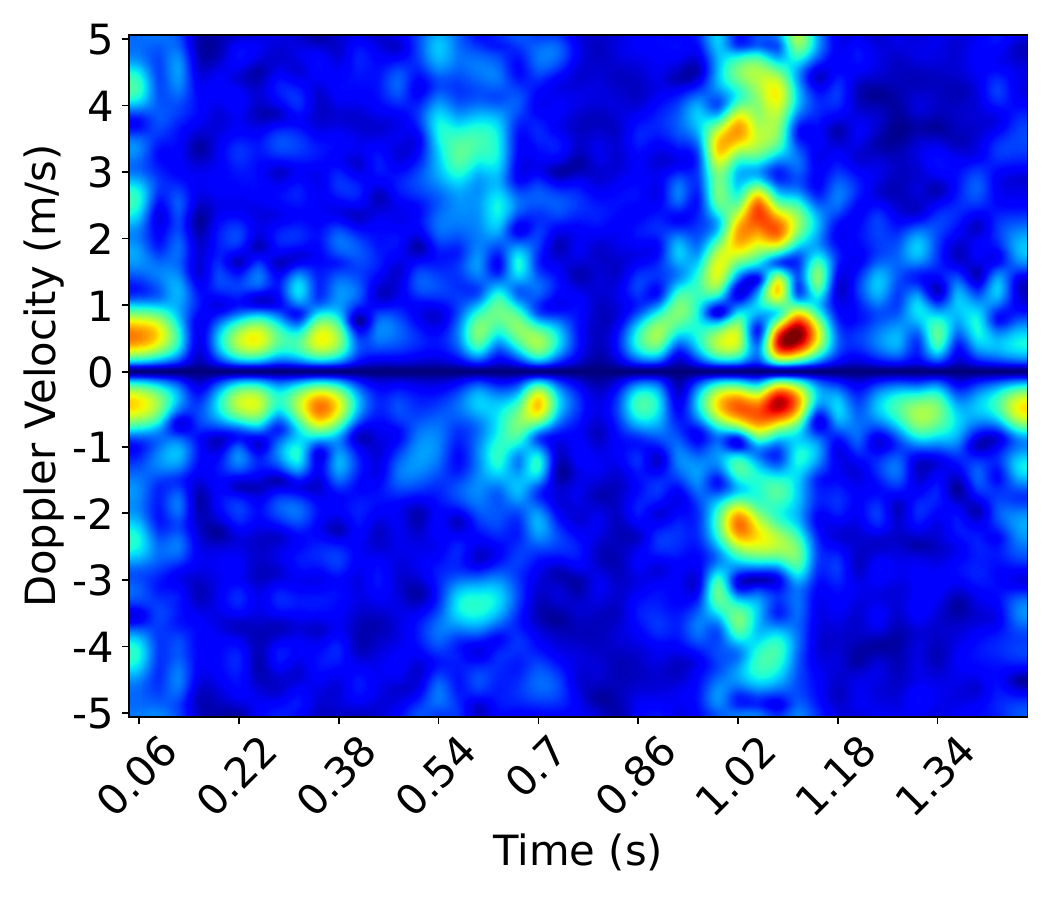}
        \label{subfig:Slide_mirror_time_doppler}}
    \end{subfloat}
\caption{Impact of artifact suppression on micro-Doppler signatures. 
(a)-(c) show ``Push \& Pull'' and (d)-(f) show ``Slide''.
(a) and (d) Raw CACC Doppler. 
(b) and (e) Differential CACC Doppler. 
(c) and (f) CACC + Delay-AoA-Filtered Doppler.
\label{fig3}}
\vspace{-0.5cm}
\end{figure*}

\subsection{AoA-Domain Filtering via Steering-Structure Mismatch}
\label{subsec:aoaestimation}

The desired motion-induced component retains an ordered AoA steering progression, whereas the residual terms produced by cyclic differencing do not share the same antenna-dependent phase structure. Exploiting this structural mismatch, we apply AoA-domain filtering at each delay bin.
Accordingly, for each delay bin $\ell$, we stack $\Delta \mathbf{W}^{(\mathrm{d})}(\ell)
\triangleq
\big[
\Delta \mathbf{W}^{(\mathrm{d})}_{1}(:,\ell),\,
\ldots,\,
\Delta \mathbf{W}^{(\mathrm{d})}_{M}(:,\ell)
\big]^{\mathsf T}
\in\mathbb{C}^{M\times K}$. 
Under the approximation that the coefficient term becomes separable across the antenna-pair index and slow time (${\gamma_m(\theta_p)}/{\gamma_m(\theta_s)} \tilde{\boldsymbol\rho}_{m,p} \approx e^{-j(m-1)\mu_p}\,\tilde{\boldsymbol\rho}_{p}$), the $M_a$-point DFT along the antenna dimension yields
\begin{align}\label{eq:aoa_fft_exact}
&\Delta \mathbf{W}^{(\mathrm{da})}(\ell)
= \mathbf F_{M_a}
\mathbf{Z}_{M\rightarrow M_a}
\Delta \mathbf{W}^{(\mathrm{d})}(\ell)/M \nonumber  \\
& \quad \approx \sum_{p=1}^{P} [\mathbf q(r_p)]_{\ell}\, \mathbf s(\mu_p) (\tilde{\boldsymbol\rho}_{p} \odot \mathbf b(f^D_{p}))^{\mathsf T}  \\ 
&\quad + \sum_{p,q=1}^{P} [\mathbf q(r_p-r_q)]_{\ell} \mathbf s^{(2)}_{p,q} \big( 
\tilde{\boldsymbol\rho}_{p} \odot \mathbf b(f^D_{p}) \odot
\tilde{\boldsymbol\rho}^{\ast}_{q} \odot \mathbf b^{\ast}(f^D_{q}) \big)^{\mathsf T}, \nonumber
\end{align}
where $\mu_p=-2\pi \Delta d/\lambda(\sin\theta_p-\sin\theta_s)$. Let $u=0,1,\dots,M_a-1$ denote the AoA-bin index with $\phi_u \triangleq {2\pi u}/{M_a}$. The spatial response of the desired term at the $u$-th AoA bin is
\begin{equation}
[\mathbf s(\mu)]_{u} = \big(1-e^{-j\mu}\big) D_M(-\mu-\phi_u) - \frac{1-e^{-jM\mu}}{M}\, e^{-j(M-1)\phi_u}. \label{eq:aoa_dirichlet}
\end{equation}
The second-order spatial response $[\mathbf s^{(2)}_{p,q}]_u$ is dominated by a scaled Dirichlet kernel $D_{M-2}\big(-(\mu_p-\mu_q+\phi_u)\big)$ together with minor boundary-induced correction terms, while its exact closed-form expression is omitted for space constraints.

The AoA transform in \eqref{eq:aoa_fft_exact} serves as a spatial filter to explicitly suppress the influence of residual second-order by-products. 
For example, under the $M=3$ configuration, the desired term follows an ordered steering vector proportional to $[1,e^{-j\mu_p},e^{-j2\mu_p}]^{\mathsf T}$, whereas the second-order residual becomes $[e^{-j3\mu_p},e^{-j\mu_p},e^{-j2\mu_p}]^{\mathsf T}$, as follows from \eqref{eq:DeltaW_general_revised}. Hence, the residual terms are not coherently focused at the target AoA bin, which reduces their impact on the subsequent processing. 

\subsection{Ambiguity-Resolved Doppler Estimation}\label{subsec:dopplerfft}


After filtering the second-order by-products in the delay-AoA domains, the signal is dominated by slow-time Doppler variations. To extract these kinematics, we apply a $K_d$-point DFT along the slow-time dimension for each delay bin $\ell$, expressed as $\Delta \mathbf W^{(\mathrm{dad})}(\ell) = \Delta \mathbf W^{(\mathrm{da})}(\ell)\,
\mathbf{Z}_{K\rightarrow K_d}^{\mathsf T}\mathbf F_{K_d}^{\mathsf T}/K$. 
Under the assumption that the reflection coefficient remains quasi-static within a CPI, such that $\tilde{\boldsymbol\rho}_{p}\approx \tilde\rho_{p}\mathbf 1_K$, the response in the Doppler domain can be expressed as: 
\begin{align}\label{eq:doppler_fft_sep}
\Delta \mathbf W^{(\mathrm{dad})}(\ell)
&\approx
\sum_{p=1}^{P}
[\mathbf q(r_p)]_{\ell}\,
\mathbf s(\mu_p)\,
\tilde\rho_{p}\,
\boldsymbol\psi(f^D_p)^{\mathsf T} \\
&\quad
+
\sum_{p,q=1}^{P}
[\mathbf q(r_p-r_q)]_{\ell}\,
\mathbf s^{(2)}_{p,q}\,
\tilde\rho_{q}^{\ast}\tilde\rho_{p}\,
\boldsymbol\psi(f^D_p-f^D_q)^{\mathsf T}, \nonumber
\end{align}
where $\nu = 0,1,\dots,K_d-1$ denotes the Doppler-bin index, and $[\boldsymbol\psi(f_D)]_{\nu}
=
D_K\left(-2\pi(f_D\Delta t+{\nu}/{K_d})\right)$. 

The resulting delay-AoA-Doppler (DAD) transform completes the multidimensional focusing of the desired motion-induced response while suppressing residual second-order by-products. As shown in \eqref{eq:doppler_fft_sep}, the desired component is concentrated around its physical delay, AoA-bin, and Doppler coordinates. Although non-integer bins introduce Dirichlet-kernel leakage, most of the target energy remains localized around $\nu^\star=\mathrm{mod}(-f_p^D\Delta t\,K_d, K_d)$, thereby preserving the physically meaningful Doppler trajectory. In contrast, the residual terms remain structurally mismatched with the desired component across the transform domains: the diagonal auto-terms ($p=q$) collapse toward near-zero delay and Doppler, whereas the off-diagonal cross-terms ($p\neq q$) are dispersed over pairwise-difference delay-Doppler coordinates and do not follow the ordered AoA steering structure required for coherent focusing. Therefore, the DAD transform yields a less ambiguous and more interpretable Doppler representation.

\subsection{AoA-Delay-Aware Micro-Doppler Aggregation}
\label{subsec:microDoppler}

Although the preceding multidimensional processing suppresses residual second-order by-products and focuses the desired response, 
the motion energy is still spread over adjacent delay-AoA bins because of the finite bandwidth and limited number of antennas in practical systems.
To construct a more stable micro-Doppler representation, we aggregate the DAD response over the delay and AoA domains. 
Let $[\Delta \mathbf{W}_{t}^{(\mathrm{dad})}(\ell)]_{u,\nu}$ denote the response at the $t$-th CPI, the $\ell$-th delay bin, the $u$-th AoA bin, and the $\nu$-th Doppler bin. The aggregated micro-Doppler feature is defined as
\begin{equation}
\mu_D(t,\nu)
\triangleq
\frac{1}{N_r M_a}
\sum_{\ell=0}^{N_r-1}\sum_{u=0}^{M_a-1}
\left|
[\Delta \mathbf{W}_{t}^{(\mathrm{dad})}(\ell)]_{u,\nu}
\right|,
\label{eq:microdoppler_agg}
\end{equation}
and collecting $\mu_D(t,\nu)$ over all $t=1,\dots,T$ and $\nu=0,\dots,K_d-1$ yields the final spectrogram $\boldsymbol{\mu}_D \in \mathbb{R}^{T\times K_d}$. 
Under the common approximation that the residual fluctuations of the per-bin magnitude responses around their local means are zero-mean, mutually uncorrelated across delay-AoA bins, and have a common variance $\sigma^2$, the averaging operation in \eqref{eq:microdoppler_agg} reduces the variance of the aggregated residual by a factor of $N_rM_a$, yielding a smoother and more stable Doppler trajectory with an improved effective bin-level SNR. 

Furthermore, the resulting signature $\boldsymbol{\mu D}$ preserves kinematic ridge patterns, making it more stable under geometric transformations induced by environmental domain shifts \cite{9731802}.
As illustrated in Fig.~\ref{fig:Slide_time_doppler}, the final constructed micro-Doppler features preserve representative kinematic structures, such as the alternating acceleration--deceleration lobes of ``Push \& Pull'' and the high-frequency concentration of ``Slide''. 
For mechanism inspection, Fig.~\ref{fig3} compares three intermediate processing variants. Raw CACC Doppler is affected by mirror ambiguity and residual second-order by-products, resulting in heavily contaminated Doppler signatures. Differential CACC Doppler suppresses the dominant mirror component, but residual second-order by-products still overlap with the desired motion-induced response, blurring the kinematic ridges. CACC + Delay-AoA-Filtered Doppler weakens residual artifacts through transform-domain filtering, but without cyclic differencing, mirror ambiguity remains and still produces paired Doppler ridges with comparable magnitudes. Compared with these partial-processing variants, the final micro-Doppler representation in Fig.~\ref{fig:Slide_time_doppler} is cleaner, less ambiguous, and exhibits more distinct motion-dependent ridge structures.

\section{Evaluation on Downstream Task}\label{section_deeplearning}

To evaluate the effectiveness of the proposed ambiguity-resolved micro-Doppler representation, we conduct downstream device-free gesture recognition experiments.
 

\subsection{Experiment Setup}


Evaluations use the Widar3.0 dataset \cite{zhang2021widar3}, which contains 105,624 effective samples of six gestures (``Push \& Pull'' (PP), ``Sweep'' (SW), ``Clap'' (CL), ``Slide'' (SL), ``Draw-O'' (DO), and ``Draw-Zigzag'' (DZ)) from 16 users across diverse domain factors (5 locations, 5 orientations, 3 environments) collected over a one-transmitter/six-receiver link at a sampling rate of 1\,kHz. Signals are processed with a 0.128\,s CPI, covering a 32\,m range window at 1\,m resolution and 8 AoA bins in the spatial-frequency domain. The aggregated profiles are temporally interpolated into $\boldsymbol{\mu D} \in \mathbb{R}^{64 \times 128}$.


By offloading random phase-offset removal, mirror suppression, and residual second-order by-product suppression to the proposed physics-informed frontend, the ResNet-50 classifier \cite{he2016deep} serves as the downstream classifier for $\boldsymbol{\mu D}$.
All baseline feature representations are formatted, resized, and evaluated under the same classifier architecture and training settings. The model is trained for 384 epochs (batch size 128) using the Adam optimizer (initial learning
rate $10^{-3}$, multi-step decay) and a frequency-inversely-weighted cross-entropy loss to handle class imbalance.
Following standard practices \cite{zhang2021widar3, li2020wihf}, we use classification accuracy as the primary evaluation metric and confusion matrices for class-wise error analysis.

\subsection{Ablation Study on Artifact Suppression}

\begin{figure}[htbp]
	\centering
	\begin{subfloat}[]{
		\includegraphics[width=0.45\linewidth,height=0.3\linewidth]{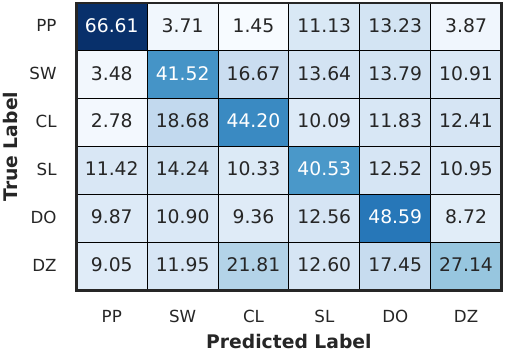}\label{subfig:CACC-Doppler_indomain}}
	\end{subfloat}
    \begin{subfloat}[]{
    \includegraphics[width=0.45\linewidth,height=0.3\linewidth]{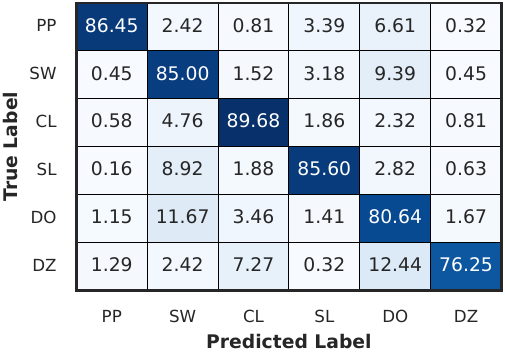}\label{subfig:Mirror-Doppler_indomain}}
	\end{subfloat}
	\begin{subfloat}[]{
		\includegraphics[width=0.45\linewidth,height=0.3\linewidth]{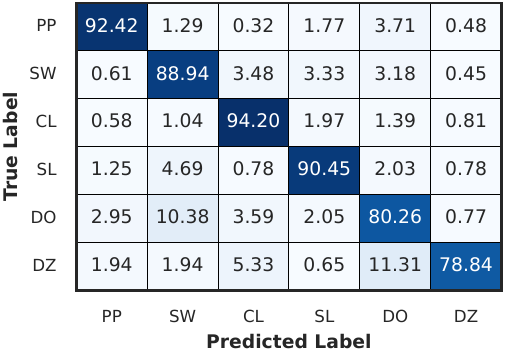}\label{subfig:NOAOA_Doppler_indomain}}
	\end{subfloat}
    \begin{subfloat}[]{
		\includegraphics[width=0.45\linewidth,height=0.3\linewidth]{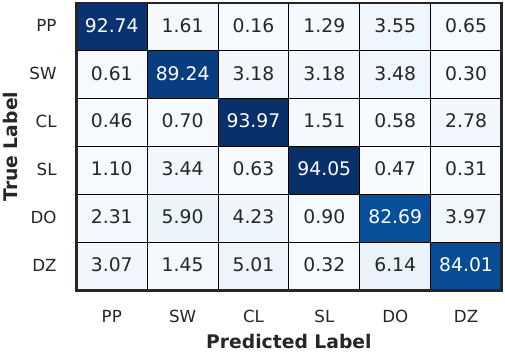}\label{subfig:MicroDopper_indomain}}
	\end{subfloat}
	\caption{Confusion matrices for in-domain gesture recognition with different features. 
    (a) Raw CACC Doppler. 
    (b) Differential CACC Doppler. 
    (c) Differential CACC + Delay-Filtered Doppler. 
    (d) Proposed micro-Doppler.}\label{fig:indomain}
    \vspace{-.4cm}
\end{figure}

Fig.~\ref{fig:indomain} presents ablation results of our proposed micro-Doppler representation against three intermediate variants, validating the effectiveness of each processing stage on the selected Location~1/Rx1 subset with 5,224 samples under a stratified 20/80 train/test split, yielding approximately 4,180 test samples. 
{\color{black} Restricting the evaluation to a single location and receiver controls cross-domain environmental variability, so that the observed performance differences can be more directly attributed to the processing modules. Moreover, using only 20\% of the data for training provides a stricter test of representation quality, showing whether the proposed modules can produce cleaner and more discriminative features even under limited training data.}
Fig.~\ref{fig:indomain}(a) shows that Raw CACC Doppler performs poorly, as the desired kinematic signatures are heavily corrupted by mirror ambiguity and cross-term interference. After dominant mirror suppression, the Differential CACC Doppler in Fig.~\ref{fig:indomain}(b) markedly improves recognition performance. However, residual second-order by-products still overlap with the desired Doppler signatures. Fig.~\ref{fig:indomain}(c) further shows that delay-only processing remains insufficient to fully suppress these residuals. In particular, 11.31\% of ``Draw-Zigzag'' samples are misclassified as ``Draw-O'', indicating persistent confusion between similar kinematic patterns. After introducing AoA-domain filtering, Fig.~\ref{fig:indomain}(d) reduces this confusion to 6.14\% and improves the diagonal concentration for ``Slide'' from 90.45\% to 94.05\%. These results show that AoA-domain filtering provides additional suppression of residual second-order by-products beyond delay-only processing.


\subsection{Cross-Domain Recognition Performance}
\label{subsec:cross_domain}

\begin{figure*}
    \centering
    \begin{subfloat}[{\color{black}Cross-location.}]{
        \includegraphics[width=.23\linewidth]{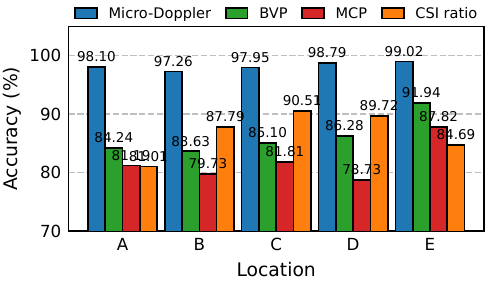}
        \label{subfig:corss_loc_acc}}
    \end{subfloat}
    \begin{subfloat}[{\color{black}Cross-orientation.}]{
        \includegraphics[width=.23\linewidth]{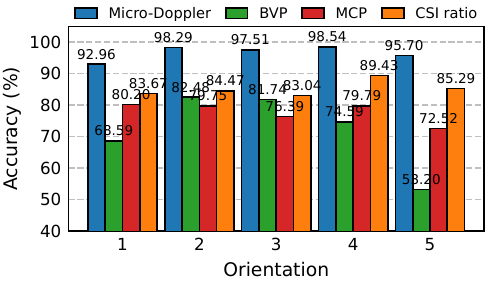}
        \label{subfig:corss_ori_acc}}
    \end{subfloat}
    \begin{subfloat}[{\color{black}Cross-environment.}]{
        \includegraphics[width=.23\linewidth]{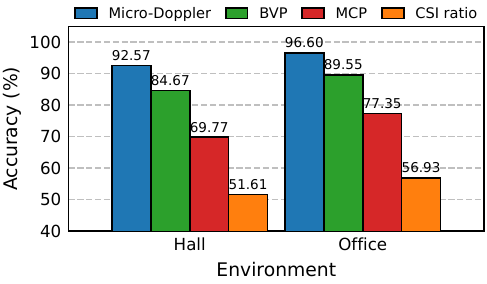}
        \label{subfig:corss_env_acc}}
    \end{subfloat}
    \begin{subfloat}[{\color{black}Cross-user.}]{
        \includegraphics[width=.23\linewidth]{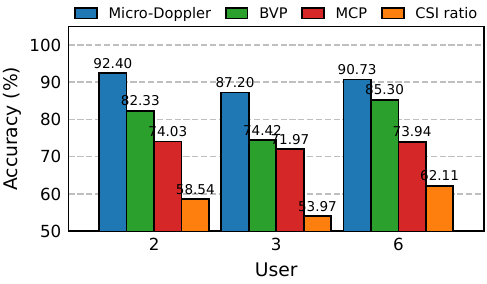}
        \label{subfig:corss_user_ac}}
    \end{subfloat}
    \caption{Accuracy comparison between the proposed micro-Doppler representation and three baselines across domains.}
    \label{fig:compair_loc}
   \vspace{-.4cm}
\end{figure*}

We next evaluate the cross-domain generalization of the proposed micro-Doppler representation on 105,624 effective samples using a leave-one-domain-out protocol across four domain factors. 
For example, in the cross-location evaluation for Location E, the model is trained on samples from Locations A--D and tested only on Location E. Fig.~\ref{fig:compair_loc} compares the proposed representation with three representative feature baselines: the body-coordinate velocity profile (BVP)~\cite{zhang2021widar3}, the motion change pattern (MCP)~\cite{li2020wihf}, and the CSI-ratio feature~\cite{2025CaoReal}. 
Overall, the proposed micro-Doppler representation achieves the highest average accuracy in all four settings among all compared features.

The gains are most pronounced under geometric domain shifts that reshape the observed micro-Doppler through bistatic projection, as shown in Fig.~4(a) and (b). 
In the cross-location setting, the proposed representation maintains accuracies above 97.2\%, whereas BVP and the CSI-ratio fluctuate within the 80\%--90\% range and MCP only reaches 78\%--85\%. In the cross-orientation setting, the proposed representation achieves an average accuracy of 96.60\%, exceeding the compared baselines by about 11 to 24 percentage points. 
This advantage arises because ambiguity resolution and residual second-order by-product suppression preserve cleaner motion-dependent ridge structures, making the representation less sensitive to changes in Doppler magnitude, ridge compression, and sign reversal. 

The proposed method also performs well under environment and user shifts, as shown in Fig.~4(c) and (d). Environment changes alter the multipath composition and background clutter statistics, while user changes introduce larger variations in motion extent, execution speed, and temporal rhythm. 
These factors can modify the duration, thickness, and energy concentration of the Doppler ridges. 
Nevertheless, the proposed representation achieves 94.59\% in the cross-environment setting, whereas the CSI-ratio feature drops to 54.27\%. In the cross-user setting, the proposed representation achieves 90.11\%, whereas BVP, MCP, and the CSI-ratio feature decrease to approximately 81\%, 73\%, and 58\%, respectively.

\section{Conclusion}
\label{section_conclusion}

This paper {\color{black}presents} an ambiguity-resolved micro-Doppler construction framework for asynchronous bistatic sensing. The key idea is to suppress the residual second-order by-products left after differential CACC through delay-AoA-Doppler domain filtering and subsequent aggregation, yielding an ambiguity-resolved micro-Doppler representation. This representation preserves motion-dependent kinematic structure while improving robustness to asynchronous phase distortion and domain shift. Experiments on the Widar3.0 benchmark {\color{black}verify} the effectiveness of the proposed multidimensional filtering and demonstrate strong cross-domain generalization, with average accuracies ranging from 90.11\% to 98.22\% across location, orientation, environment, and user. These results indicate that explicit suppression of residual second-order by-products enables robust micro-Doppler construction in asynchronous bistatic ISAC systems.

\bibliographystyle{IEEEtran}

\bibliography{ref}

\newpage

\end{document}